\documentclass[11pt]{article}
\usepackage{amsmath,amsthm,amssymb}
\usepackage{iftex}
\usepackage{hyperref} %use of hypertext in the document
\usepackage{xcolor}

\ifXeTeX
\usepackage{fontspec}
\usepackage{xltxtra}

\fi
\numberwithin{equation}{section} % up to section

\newtheorem{lma}{Lemma}

\newtheorem{Def}{Definition}

\newtheorem*{auxthm}{Auxiliary Theorem}

\newtheorem*{Mthm}{Main Theorem}

\usepackage{setspace} %footnote single space

\newcommand{\pfa}{{\noindent\bfseries\itshape Proof of the Auxiliary Theorem. }}
\newcommand{\pfm}{{\noindent\bfseries\itshape Proof of the Main Theorem. }}

\newcommand{\norm}[1]{\left\Vert #1 \right\Vert}

\begin{document}

\title{Equilibria in a large production economy \\  
with an infinite dimensional commodity space \\ and  price dependent preferences\thanks{We thank two anonymous referees for their comments and suggestions. 
%We are also grateful to M. Ali Khan for encouragement.  Errors are, of course, solely ours.  
This research was supported by Hankuk University of Foreign Studies Research Fund.} 
}

\author{Hyo Seok Jang\footnote{%
Department of Mathematical Sciences, Seoul National University, Seoul, 
Korea. Email: hyoseok.jang@snu.ac.kr} \;\;
Sangjik Lee\footnote{% 
Division of Economics, Hankuk University of Foreign Studies, Seoul,  Korea. Email: slee@hufs.ac.kr.} 
} 

\maketitle

\vspace{5mm}
{\linespread{1}
\begin{abstract}
\noindent
%We extend Greenberg et al. \cite{gsw} to a production economy with infinitely many commodities and prove the existence of a competitive equilibrium for the economy. We employ a saturated measure space for the set of agents and apply recent results for an infinite dimensional separable Banach space such as Lyapunov's convexity theorem and an exact Fatou's lemma to obtain the result.
We prove the existence of a competitive equilibrium in a production economy with infinitely many commodities and a measure space of agents whose preferences are price dependent. We employ a saturated measure space for the set of agents and apply recent results for an infinite dimensional separable Banach space such as Lyapunov's convexity theorem and an exact Fatou's lemma to obtain the result.

\vspace{5mm}
\noindent
\textbf{JEL Classification Numbers:} C62, D51.

\vspace{5mm} 
\noindent 
\textbf{Keywords:} Separable Banach space, Saturated measure space, Price dependent preferences, Lyapunov's convexity theorem, Fatou's lemma 
\end{abstract}
}

\section{Introduction} % (fold)
\label{sec:introduction}
The purpose of this paper is to prove the existence of a competitive equilibrium in a production economy with infinitely many commodities and a measure space of agents whose preferences are price dependent.
In a seminal paper, Aumann \cite{aumann1966} demonstrated the existence of a competitive equilibrium for an exchange economy with a finite dimensional commodity space and a continuum of agents modeled as an atomless finite measure space by utilizing Lyapunov's convexity theorem to dispense with convex preferences. Aumann's model in \cite{aumann1966} was generalized to allow incomplete preferences by Schmeidler \cite{schmeidler1969} and to include production by Hildenbrand \cite{hildenbrand70}.

As Shafer\cite{shafer1974}\footnote{This was drawn to our attention by an anonymous referee.}  and Balasko \cite{balasko2003a} pointed out, %price dependent preferences have long been recognized. 
price dependent preferences have been traditionally explained by consumers taking relative prices as an indication of quality. In addition, we see other applications of price dependent preferences in the literature: Shafer \cite{shafer1974} showed the possibility of relating price dependent preferences to non-transitive preferences and Balasko \cite{balasko2003b} demonstrated the equivalence of a temporary financial equilibrium model with an Arrow-Debreu economy where preferences are price dependent. 

%The existence of a competitive equilibrium in a large economy with price dependent preferences and a finite number of commodities was first proved by Greenberg et al. \cite{gsw}.

Greenberg et al. \cite{gsw} first proved the existence of a competitive equilibrium in a large economy with price dependent preferences and a finite number of commodities.
In \cite{gsw}, the authors considered a large production economy with non-convex preferences. They reformulated the production economy as a three-person game and applied Debreu's social equilibrium existence theorem to obtain a Walrasian equilibrium. 
In their proof, they applied Lyapunov's convexity theorem and Fatou's Lemma in several dimensions. 
In order to utilize Fatou's lemma, Greenberg et al. \cite{gsw} assumed the compactness of the consumption sets, which differs from Aumann's original model. Liu \cite{liu} dealt with a coalition production economy based on Greenberg et al. \cite{gsw}.

For infinite dimensional commodity spaces,  
Khan and Yannelis \cite{ky91} considered a large exchange economy and showed the existence of a competitive equilibrium.  
In \cite{ky91}, the commodity space is an ordered separable Banach space whose positive cone has a non-empty interior. Until recently, Lyapunov's convexity theorem and an exact Fatou's lemma for an infinite dimensional separable Banach space were not available. Therefore, the authors had to impose the assumption of convex preferences. They relied on the weak compactness of feasible allocations to extract a convergent subsequence of competitive equilibria for truncated subeconomies to obtain the existence of a Walrasian equilibrium. 
Now that the necessary mathematical tools are at hand, it is natural to ask as to whether equilibrium existence results for a large economy with an infinite dimensional commodity space, non-convex preferences and price externalities are available. We give a positive answer in this paper. 

Saturated or super-atomless measure spaces have played an important role in recent mathematical economics. Podczeck \cite{pod2008} and Sun and Yannelis \cite{sy} successfully proved the convexity of Bochner integrals of an infinite dimensional separable Banach space valued correspondence on a saturated measure space. Based on saturated measure spaces, Khan and Sagara \cite{ks2013} proved Lyapunov's convexity theorem for vector measures taking values in an infinite dimensional separable Banach space and Greinecker and Podczeck \cite{gp2013} also showed it. Khan and Sagara \cite{ks2014} established an exact Fatou's lemma for an infinite dimensional separable Banach space. Khan et al. \cite{kss2016} proved an exact Fatou lemma for Gelfand integrals which was also established via Young measures by  Greinecker and Podczeck \cite{gp2017}. 
These results have already been applied to general equilibrium theory in several papers; see Khan and Sagara \cite{ks2016, ks2017}, Khan and Suzuki \cite{ksu2016} and Lee \cite{lee}.  In \cite{ks2017}, the authors  emphasized the importance of saturated measures by saying that ``the significance of the saturation property lies in the fact that it is necessary and sufficient for the weak/weak* compactness and the convexity of the Bochner/Gelfand integral of a multifunction as well as the Lyapunov convexity theorem in separable Banach spaces/their dual spaces.''

In this paper, we consider a large production economy whose commodity space is that of Khan and Yannelis \cite{ky91} and whose agents have non-convex and price dependent preferences, similar to Greenberg et al. \cite{gsw}.
We employ a saturated measure space of agents and hence, we can utilize the convexity of a Bochner integral of a Banach space valued correspondence, Lyapunov's convexity theorem, and the exact Fatou's lemma for an infinite dimensional Banach space. With these new results, we are able to relax the convexity of preferences and production sets, and apply Debreu's social equilibrium existence theorem. Moreover, we can obtain a competitive equilibrium as the limit of a sequence of competitive equilibria for truncated subeconomies. 
We dispense with the uniform compactness assumption on the consumption sets 
and production sets, which was used in \cite{gsw} and in \cite{liu}. %

The paper proceeds as follows: Section 2 contains notations and definitions. We present our model in Section 3, and our main and auxiliary results are in Section 4. The proof of the auxiliary result is in Section 5 followed by the proof of the main theorem in Section 6.  Section 7 concludes the paper with our remarks.

\section{Notation and Definitions}
Let \(X, Y\) be topological spaces. A set-valued function or a correspondence \(F\) from \(Y\) to the family of non-empty subsets of \(Y\) is called \emph{upper semicontinuous} if the set \(\{x:X: F(x)\subset V\}\) is open in \(X\) and  said to be \emph{lower semicontinuous} if the set \(\{x:X: F(x)\cap V\neq \emptyset\}\) is open in \(X\) for every \(V\) of \(Y\). When \(Y\) is a Banach space, \(F\) is \emph{norm upper semicontinuous} if the set \(\{x:X: F(x)\subset V\}\) is open in \(X\) for every norm open subset \(V\) of \(Y\). And \(F\) is called \emph{weakly upper semicontinuous} if the set \(\{x:X: F(x)\subset V\}\) is open in \(X\) for every weakly open subset \(V\) of \(Y\). We say that \(F\) is \emph{norm lower semicontinuous} if the set \(\{x:X: F(x)\cap V\neq \emptyset\}\) is open in \(X\) for every norm open subset \(V\) of \(Y\) and \(F\) is said to be \emph{weakly lower semicontinuous} if the set \(\{x:X: F(x)\cap V\neq \emptyset\}\) is open in \(X\) for every weakly open subset \(V\) of \(Y\).

Let \((T,{\mathcal T},\mu)\) be a finite measure space and \(E\) be a Banach space. 
A measurable function \(f:(T,{\mathcal T},\mu) \to E\) is said to be \emph{Bochner integrable} if there exists a sequence of simple functions \(\{f_n\}_{n\in \mathbb{N}}\) such that 
\begin{equation}
\lim_{n\to \infty}\int_T \norm{f_n(t) - f(t)} d\mu=0
\end{equation}
where \(\mathbb{N}\) denotes the set of natural numbers. 
For each \(S\in {\mathcal T}\) the integral is defined to be \(\int_S f(t)d\mu = \lim_{n\to \infty}\int_S f_n(t)d\mu\).   
Denote by \( L^1(\mu,E) \) the space of (the equivalence classes of) \(E\)-valued Bochner integrable functions \( f:T\to E \) normed by \( \norm{f }_1=\int_T \norm{ f(t) } d\mu \). 

The \emph{weak upper limit} of a sequence \(\{S_n\}\) of subsets in \(E\) is defined by
\begin{equation}
 w\text{-Ls}\; S_n = \{x\in E: \exists \{x_{n_k}\} \text{ such that } x=w\text{-}\lim x_{n_k}, x_{n_k}\in S_{n_k},\text{ for all } k\in \mathbb{N}\}
\end{equation}
where \(\{x_{n_k}\}\) is a subsequence of a sequence \(\{x_n\}\) and \(w\text{-}\lim_n x_{n_k}\) denotes the weak limit point of \(\{x_{n_k}\}\). 
  
A correspondence \(F:T\to 2^{E}\) is said to be \emph{measurable} if for every open subset \(V\) of \(E\), the set \(\{t\in T: F(t)\cap V\neq \emptyset\}\in {\mathcal T}\). The correspondence \(F\) is said to have a  \emph{measurable graph} if its graph \(G_F=\{(t,x)\in T\times E: x\in F(t)\}\) belongs to the product \(\sigma\)-algebra \({\mathcal T}\otimes {\mathcal B}(E,w)\), where \({\mathcal B}(E,w)\) denotes the Borel \(\sigma\)-algebra of \(E\) generated by the weak topology. If correspondences from \(T\) to \(E\) are closed valued, measurability and graph measurability are equivalent when \((T,\mathcal{T},\mu)\) is complete and \(E\) is separable.\footnote{See Theorem 8.1.4 in \cite{auf}.} A measurable correspondence \(F:T\to 2^{E}\) is \emph{integrably bounded} if there exists a real-valued integrable function \(h\) on \((T,{\mathcal T},\mu)\) such that  \(\sup\{\norm{x} :x\in F(t)\}\leq h(t)\) for almost all \(t\in T\).

A measurable function \(f\) from \((T,{\mathcal T},\mu)\) to \(E\) is called a \emph{measurable selection} of the correspondence \(F\) if \(f(t)\in F(t)\) for almost all \(t\in T\). By Aumann's measurable selection theorem in \cite{aumann1969}, if \((T,{\mathcal T},\mu)\) is a complete finite measure space, \(F\) has a measurable graph, and \(E\) is separable, then \(F\) has a measurable selection.  We denote by \({\mathcal S}^1_F\) the set of all \(E\)-valued Bochner integrable selections for the correspondence \(F\), i.e., 
\(
{\mathcal S}_F^1=\{f\in L^1(\mu,E):f(t)\in F(t)\; \text{a.e. }t\in T\}.
\)
When \(F\) is also integrably bounded, it admits a Bochner integrable selection so that \({\mathcal S}^1_F\) is non-empty. 
The integral of the correspondence \(F\) is defined by
\begin{equation}
\int_T F(t)d\mu=\{\int_T f(t)d\mu:f\in {\mathcal S}_F^1\}.
\end{equation}

A sequence of correspondences \(\{F_n\}\) from \(T\) to \(E\) is said to be \emph{well-dominated} if there exists an integrably bounded and weakly compact-valued correspondence \(\phi: T\to 2^{E}\) such that \(F_n(t)\subset \phi(t)\) a.e. \(t\in T\) for each \(n\).

Let \(E\) be an ordered Banach space equipped with ordering \(\geq\) such that the positive cone \(E_+=\{x\in E: x\geq 0\}\) of \(E\) is closed. For \(x,y\in E\), \(x>y\) means \(x-y\in E_+\) and \(x\neq y\). 
%A function \(f:E\to\mathbb{R}\) is said to be strictly increasing, if, for \(x\) and \(y\in E\), \(x>y\) implies \(f(x)> f(y)\).
%On the other hand, since \( E \) is a topological vector space, we can define the \emph{topological cone} \( C( A ):=\{ \lambda x \in E : x \in A, 0 \leq \lambda \leq 1\} \) over any subset \(A \) of \(E\). 
We denote by  \(E^*\) the dual space of \(E\), i.e., the space of all continuous linear functionals from \(E\) into \(\mathbb{R}\). For \(x\in E, p\in E^*\), we write \(p\cdot x\) for the value of \(p\) at \(x\).  We denote by \(E^*_+\) the dual cone of \(E_+\), i.e., \(E^*_+=\{p\in E^*:p\cdot x\geq 0\;\text{ for all } x\in E_+\}\). We denote by \(\mathcal{B}(E^*,w^*)\) the Borel \(\sigma\)-algebra of \(E^*\) generated by the weak* topology.
For any set \(A\) in \(E\), \(\text{cl}{ A }\) stands for the norm closure of \(A\) and \(\text{co}A\) for the convex hull of \(A\).

Let \((T,\mathcal{T},\mu)\) be a finite measure space. Denote by \(L^1(\mu)\) the the space of (\(\mu\)-equivalence classes of) real valued integrable functions on \(T\). Let \(\mathcal{T}_S=\{A\cap S|A\in \mathcal{T}\}\) be the sub-\(\sigma\)-algebra of \(\mathcal{T}\) restricted to \(S\in \mathcal{T}\) and \(\mu_S\) be a restriction of \(\mu\) to \(\mathcal{T}_S\). We write \(L^1_S(\mu)\) for the vector subspace of \(L^1(\mu)\) which consists of each function in \(L^1(\mu)\) restricted to \(S\).

\begin{Def}
A finite measure space \((T, \mathcal{T},\mu)\) is saturated if \(L^1_S(\mu)\) is non-separable for every \(S\in \mathcal{T}\) with \(\mu(S)>0\). 
\end{Def}

A saturated measure space is also called ``super-atomless'' in Podczeck \cite{pod2008}. Other equivalent definitions for saturation are available in the literature; see \cite{fjrdkler2002}, \cite{fremlin2012}, \cite{hk1984}, \cite{keislersun2009}, and \cite{pod2008}. 
As mentioned in Khan and Sagara \cite{ks2017}, ``a germinal notion of saturation already appeared in \cite{kakutani1944, maharam},'' and Kakutani \cite{kakutani1944} constructed a non-separable extension of the Lebesgue measure space which can be seen as a saturated extension of the Lebesgue  interval.\footnote{We are grateful to an anonymous referee for drawing our attention to Kakutani \cite{kakutani1944}.} 
Examples of saturated measure spaces include the product spaces of the form \([0,1]^\kappa\) and \(\{0, 1\}^\kappa\), where \(\kappa\) is an uncountable cardinal, \([0, 1]\) is endowed with the Lebesgue measure and \(\{0, 1\}\) the fair coin flipping measure. The cardinalities of these two examples are greater than the continuum.    
Podczeck \cite{pod2008} constructed a saturated  measure structure on the unit interval by ``enriching'' the Lebesgue \(\sigma\)-algebra. 
Thus, as is pointed out in \cite{pod2008}, when we have a saturated measure space of agents, the cardinality of the set of agents is not necessarily larger than the continuum. 
%As mensioned in Khan and Sagara \cite{ks2017}, ``a germinal notion of sautration already appeared in \cite{kakutani1944, maharam},'' and Kakutani \cite{kakutani1944} established an earlier version of the saturated extension of the Lebesgue unit interval.\footnote{We are grateful to an anonymous referee for drawing our attention to Kakutani \cite{kakutani1944}.} 

\section{The Model} \label{The Model}
The commodity space \(E\) is an ordered separable Banach Space with an interior point \(v\) in \(E_+\).\footnote{The examples of this space include \(C(K)\), the set of bounded continuous functions on a Hausdorff compact metric space \(K\) equipped with sup norm and a weakly compact subset of \(L_\infty(\mu)\) where \(\mu\) is a finite measure.} 
For the space of agents, we employ a complete probability space \((T, {\mathcal T}, \mu)\) which is saturated. 
Let \(X\) be a correspondence from \(T\) to \(E_+\). The consumption set of agent \(t\in T\) is given by \(X(t)\subset E_+\). The initial endowment of each agent is given by a Bochner integrable function \(e:T\to E\) where \( e(t)\in X(t)\) for all \(t\in T\). The aggregate initial endowment is \(\int_T e(t)d\mu\). Let \(Y\) be a correspondence from \(T\) to \(E\). The production set of agent \(t\) is given by \(Y(t)\subset E\). A price is \( p\in E^*_+\backslash \{0\}\). Let \(\Delta=\{p\in E^*_+\backslash\{0\}: p\cdot v =1\}\) be the price space. 
Then by Alaoglu's theorem, \(\Delta\) is weak* compact.
Let \(\mathcal{E} =[(T,{\mathcal T}, \mu), (X(t),Y(t),U_t,e(t))_{t\in T}]\) be a production economy where \(U_t:X(t)\times \Delta \to \mathbb{R}\) represents agent \(t\)'s utility function. We also write \(U(t,x,p) = U_t(x,p)\) for \(t\in T\), \(x\in X(t)\) and \( p\in \Delta\).
An allocation for \(\mathcal{E}\) is a Bochner integrable function \(f:T\to E_+\) such that \(f\in \mathcal{S}_X^1\) and a production plan is a Bochner integrable function \(g:T\to E\) such that \(g\in \mathcal{S}_Y^1\). 
 The budget set of agent \(t\) at a price \( p\in \Delta\) is \(B(t,p)=\{x\in X(t):p\cdot x \leq p \cdot e(t) + \text{max }p\cdot Y(t) \}\).

A competitive equilibrium for \(\mathcal{E}\) is a triple of a price \(p\), an allocation \(f\) and a production plan \(g\) such that 

\begin{enumerate}
	\item \(p\cdot f(t) \leq p\cdot e(t) + p\cdot g(t)\) for almost all \(t\in T\),
	\item \(\int_T f(t)d\mu\leq \int_T e(t)d\mu + \int_T g(t) d\mu\),
	\item for any \(x\in X(t)\), \(U_t(x,p)> U_t(f(t),p)\) implies that \(p\cdot x> p\cdot e(t) + p\cdot g(t)\) for almost all \(t\in T\),
	\item  \(p\cdot g(t) = \text{max } p\cdot Y(t)\) for almost all \(t\in T\).
\end{enumerate}

We assume that the production economy \(\mathcal{E}\) satisfies the following assumptions:

\begin{itemize}
\item[A.1] \(X(t)\) is non-empty, closed, convex, integrably bounded and weakly compact for all \(t\in T\).
          
\item[A.2] \(Y(t)\) is non-empty, closed, integrably bounded and weakly compact for all \(t\in T\).       

\item[A.3] There is an element \(\eta(t)\in X(t)\) such that 
           \(e(t) - \eta(t)\) is in the norm interior of \(E_+\) for all \(t\in T\). 
		   
\item[A.4] (i) \(U_t: X(t)\times \Delta \to \mathbb{R}\) is a jointly continuous  function on \(X(t)\times \Delta \) for all \(t\in T\) where \(X(t)\) is equipped with the weak topology and \(\Delta\)  with the weak* topology. 
 (ii) If \(x\in X(t)\) is a satiation point for \(U_t(\cdot,p)\), then \(x\geq e(t) + y\) for any \(y\in Y(t)\); if \(x\in X(t)\) is not a satiation point for \(U_t(\cdot,p)\), then \(x\) belongs to the weak closure of the set \(\{x'\in X(t): U_t(x',p)> U_t(x,p)\}\) for every \(p\in \Delta\).

\item[A.5] \(U\) is jointly measurable with respect to \(\mathcal{T}\otimes \mathcal{B}(E,w)\otimes \mathcal{B}(E^*,w^*)\). 
%on \(G_X\times \Delta \) where \( G_X=\{(t,x)\in T\times E: x\in X(t)\}\).
      
\item[A.6] the correspondence \(X:T\to 2^{E}\) has a measurable graph,
      i.e.,
      \(G_X\in {\mathcal T}\otimes {\mathcal B}(E,w) \).
             
\item[A.7] the correspondence \(Y:T\to 2^{E}\) has a measurable graph,
      i.e.,
      \(G_Y=\{(t,y)\in T\times E: y\in Y(t)\}\in {\mathcal T}\otimes {\mathcal B}(E,w) \).
       
\item[A.8] \(\mathbf{0}\in Y(t)\) for all \(t\in T\) where \(\mathbf{0}\) is the zero vector of \(E\).           
\end{itemize}

%\noindent\textbf{Remark} 
In A.1 and A.2, we assume that both the consumption sets and the production sets are weakly compact. Although these assumptions seem strong, the  weakly compact consumption set assumption was employed in Khan and Yannelis \cite{ky91}, Podczeck \cite{pod1997} and Khan and Sagara \cite{ks2017}.\footnote{Since these three papers dealt with exchange economies, the production sets are irrelevant.}
With this assumption, Khan and Yannelis \cite{ky91} made the set of feasible allocations weakly compact, Podczeck \cite{pod1997} obtained a weakly compact-valued demand correspondence, and Khan and Sagara \cite{ks2017} were able to invoke the exact Fatou's lemma for an infinite dimensional separable Banach space.
We use the weak compactness assumption to apply the exact Fatou's lemma for our results. 
%In A.3, we assume an additional condition that \(\eta(t)\) lies in the norm closed convex hull of the norm compact set \(K\). This extra assumption is needed for the approximation approach employed for the existence proof. 
A.4 (ii) is imposed in \cite{ks2017, lee, pod1997} and the second part plays a similar role to the ``local nonsatiation'' assumption.

%We can replace the assumption of the weak compactness of \(X(t)\)  and \(Y(t)\)  with  the weak compactness of the relevant sets of allocations and production plans. We will discuss this later in Section \ref{conclusion}. 

\section{Results}
The following theorem is our main result:

\begin{Mthm}
\label{mainresult}
Suppose that the production economy \(\mathcal{E}\) satisfies A.1-A.8. Then there exists a competitive equilibrium for \(\mathcal{E}\).
\end{Mthm}

The proof of the Main Theorem is provided in Section \ref{proof_main}. 
As is well known, for \(x\in E\) and \(p\in \Delta\) the bilinear map \((p,x)\mapsto p\cdot x\) is not jointly continuous if \(E\) is equipped with the weak topology and \(\Delta\) with the weak* topology. But 
when \(E\) is equipped with the norm topology, the bilinear map is continuous.\footnote{See Aliprantis and Border \cite{ab}  pp. 241-242.} 
To utilize this property, we modify A.1 and A.2:

\begin{itemize}
   \item[A.1\('\)] \(X(t)\) is non-empty, closed, convex, integrably bounded and norm compact for all \(t\in T\).
   \item[A.2\('\)] \(Y(t)\) is non-empty, closed, integrably bounded and norm compact for all \(t\in T\). 
\end{itemize}

We now introduce the following auxiliary result:

\begin{auxthm}
  \label{auxresult}
  Suppose that the production economy \( \mathcal{E} \) satisfies A.1\('\), A.2\('\) and  A.3-A.8. Then there exists a competitive equilibrium for \( \mathcal{E} \).
\end{auxthm}

We provide the proof of the Auxiliary Theorem in Section \ref{proof_aux}. We follow the idea of \cite{gsw} for the proof of the Auxiliary Theorem. Greenberg et al. \cite{gsw} applied Debreu's \cite{debreu} social equilibrium result to prove the existence of a competitive equilibrium. 
 
We introduce a 3-person game \(\Gamma\) which consists of three sets \(K_1, K_2, K_3\), and three correspondence \(A_1:K_2\times K_3\to 2^{K_1}, A_2:K_1\times K_3\to 2^{K_2}, A_3:K_1\times K_2\to 2^{K_3}\), and three functions \(u_i:K_1\times K_2\times K_3 \to \mathbb{R}\) \((i=1,2,3)\). 
Let \(I=\{1,2,3\}\) and let \(K_{-i}= \Pi_{j\neq i} K_j\;(i,j\in I) \). We write \(k_{i}\) for an element in \(K_{i}\) and \(k_{-i}\) for \(K_{-i}\).

An equilibrium for \(\Gamma\) is \(k^*\in K_1\times K_2\times K_3 \)  such that for all \(i\in I\)
	\begin{equation}
        k_i^* \in \text{argmax }_{k_i\in A_i(k_{-i}^*)  }  u_i(k_i,k_{-i}^*).
	\end{equation}

The following lemma is Debreu's \cite{debreu} social equilibrium theorem for a Banach space.

\begin{lma}\label{lma_debreu}
  Let \(\Gamma\) be a 3-person game and suppose \(\Gamma\) satisfies, for \(i\in I\), 
  \begin{itemize}
  \item[(i)] \(K_i\) is a non-empty, convex, and compact subset of a Banach space;
  \item[(ii)] \(A_i\) is continuous, non-empty, closed and convex valued;
  \item[(iii)] \(u_i\) is continuous and quasi-concave on \(K_i\). 
  \end{itemize}
  Then \(\Gamma\) has an equilibrium.
\end{lma}

\begin{proof}
By applying a standard argument to our Banach space, we can have the result.
\end{proof}

Based on Lemma \ref{lma_debreu}, we will prove the Auxiliary Theorem. 
Toward this end, we specify our \(\Gamma\). Without loss of generality, we assume the values of \(U_t\) are contained in \([0,1]\) for all \(t\in T\). 
Let \(K_1= \Delta\), \(K_2=\int_T X(t)d\mu \times [0,1]\), and \(K_3=\int_T Y(t)d\mu\).  
For \( p\in K_1, (x,\alpha)\in K_2\) and \(y\in K_3\), let
\(A_1((x,\alpha),y)  = K_1 \), 
\(A_2(p,y)  = \{(x,\alpha)\in K_2: \exists f\in \mathcal{S}_X^1 \text{ such that } x=\int_T f(t)d\mu, f(t)\in B(t,p) \text{ a.e }t\in T, \alpha = \int_T U_t(f(t),p)d\mu\}\),
\(A_3(p,(x,\alpha))  = K_3,\)
and 
\begin{equation}
 u_1(p,(x,\alpha),y)  = p\cdot (x -\int_T e(t)d\mu -y),\;\;
u_2(p,(x,\alpha),y)  = \alpha, \;\;
 u_3(p,(x,\alpha),y)  = p\cdot y.
\end{equation}

\begin{lma} 
\label{int_sy}
Under A.1\('\) and A.2\('\), \(\int_T X(t) d\mu\) and \(\int_T Y(t) d\mu\) are norm compact and convex. 
\end{lma}
\begin{proof}
	By appealing to Proposition 1 in Sun and Yannelis \cite{sy}, we have the results.
\end{proof}

\begin{lma}
   \label{Bp}
	\(B(t,p)\) is a non-empty and continuous correspondence in \(p\) 
	when \(X(t)\)  and \(Y(t)\) are norm compact and \(\Delta\) is weak* compact.
\end{lma}
\begin{proof}
By A.8, it is clear that \(\text{max }p\cdot Y(t)\geq 0\). Then \(\eta(t)\in B(t,p) \) for any \( p\in \Delta\). Therefore, \(B(t,p)\) is non-empty.

Let \(\psi_t:\Delta \to \mathbb{R}\) be a function defined by 
	 \(
 \psi_t(p)=\text{max }_{y\in Y(t)}\;p\cdot y.
     \)
By Berge's theorem, \(\psi_t(p) \) is continuous in \( p\). 
We define a function \(z_t:\Delta\to\mathbb{R}\) by 
 \begin{equation}
  z_t(p) = p\cdot e(t) + \text{max } p\cdot Y(t) = p\cdot e(t) + \psi_t(p).
 \end{equation}
Clearly, \(z_t(p) \) is continuous in \(p\). The budget correspondence can be rewritten as \(B(t,p) = \{x\in X(t): p\cdot x\leq z_t(p)\}\). By A.3 and A.8, \( z_t(p) > 0\) for all \( p\in \Delta\). Then a standard argument can be adopted to show that \(B(t,p)\) is continuous in \(p\).
\end{proof}

The following is the exact Fatou's lemma for Banach spaces proved by Khan and Sagara \cite{ks2014}. 
\begin{lma}[Theorem 3.5 in \cite{ks2014}] %Khan and Sagara (2014)]
\label{fatou}
Let \( (T,\mathcal{T},\mu) \) be a complete saturated finite measure space and \(E\) be a Banach space. If \( \{f_n\}\) is a well-dominated sequence in \(L^1(\mu,E) \), then there exists \(f\in L^1(\mu,E) \) such that
\begin{enumerate}
 \item[(i)] \(f(t)\in w\text{-Ls}\{f_n(t)\}\) a.e. \(t\in T\),
 \item[(ii)] \(\int f d\mu \in w\text{-Ls}\{\int f_n d\mu\}\).
\end{enumerate}
\end{lma}

\begin{lma}
\label{A_property}
Under A.1\('\) and A.2\('\), \(A_i\) is continuous, non-empty, closed and convex valued for \(i=1,2,3\).  
\end{lma}
\begin{proof}
%\subsection{Proof of Lemma \ref{A_property}}
We adopt the idea of the proof from \cite{gsw}. 
It is clear that \(K_1=\Delta\) is non-empty and convex. By Alaoglu's theorem, it is weak* compact and thus, weak* closed.  It follows that \(A_1\) is non-empty, closed and convex valued. 
From A.8, \(\mathbf{0}\in \int_T Y(t)d\mu\) and thus \(K_3=\int_T Y(t)d\mu\) is non-empty. By Lemma \ref{int_sy}, \(\int_T Y(t) d\mu\) is convex and norm compact and thus, norm closed. Hence, \(A_3\) is non-empty, closed and convex valued. Clearly, \(A_1\) and \(A_3\) are continuous.
  
We now turn to \( A_2 \). Since the initial endowment map \(e(t)\in B(t,p) \), \(A_2\) is non-empty. 
Note \(\int_T e(t)d\mu\in \int_T X(t)d\mu\) for all \(t\in T\) and \(\int_T U_t(e(t),p)d\mu\in [0,1]\). Thus, \(K_2=\int_T X(t)d\mu \times [0,1]\) is non-empty. By Lemma \ref{int_sy}, \(\int_T X(t)d\mu\) is norm compact and convex. It follows that \(K_2\) is compact and convex.  

We show the value of \( A_2\) is closed. We need to show \( (x,\alpha) \in A_2(p,y) \) when \(x_n\to x\) in norm and \(\alpha_n\to \alpha\) such that  \((x_n,\alpha_n)\in A_2(p,y)\) for all \(n\).  
Then there exists a sequence \(\{f_n\}\subset \mathcal{S}_X^1\) such that \(x_n = \int_T f_n(t)d\mu\) and \(\alpha_n = \int_T U_t(f_n(t),p)d\mu\) with \(f_n(t)\in B(t,p)\) for all \(n\). %Clearly, \(\{f_n\}\) is well-dominated. 
By virtue of A.1\('\),  \(\{f_n\}\) is well-dominated.
We can appeal to Lemma \ref{fatou} to have \(f\in L^1(\mu,E)\) such that \(f(t)\in X(t), \) \( f(t) \in  w\text{-Ls} \{ f_n(t) \}\) for a.e. \( t \in T, \) and \(\int_T f(t)d\mu\in w\text{-Ls}\{\int f_n d\mu\}\). Thus we can extract a subsequence from \(\{f_n \}\) (which we do not relabel) such that \( f_n(t) \to f(t) \) weakly for a.e. \( t \in T \) and  \(\int_T f_n(t)d\mu \to \int_T f(t)d\mu\) weakly. Since \( B(t,p)\) is norm compact and \(f_n(t)\in B(t,p)\) for all \(n\), it follows \(f(t)\in B(t,p)\). 
The weak limit \( \int_T f(t)d\mu\) of the subsequence of \( \{ x_n \} \) must be equal to the norm limit \( x\) of the whole sequence \( \{ x_n\}. \)
Because \( U_t(\cdot,p)\) is weakly continuous, \( U_t(f_n(t),p) \to U_t(f(t),p)\)  for a.e. \( t \in T \). 
On the other hand, let  \( g_n(t)=U_t( f_n(t),p) \). Then from the boundedness of \( U \), the sequence of functions \(\{g_n\}\) is well-dominated. 
Lemma \ref{fatou} implies that there exists \(g\in L^1(\mu)\) such that \( g_n(t) \to g(t) \) for a.e. \( t \in T \) and \( \alpha_n=\int_T g_n(t)d\mu \to \int_T g(t)d\mu\) up to subsequence. 
Hence \( g(t)=U_t(f(t),p)\)  for a.e. \( t \in T \) and \( \alpha=\lim_n \alpha_n =\int_T g(t)d\mu=\int_T U_t(f(t),p) d\mu. \)

Next, we show the upper semicontinuity of \(A_2\). Since \(K_2\) is compact, in order to prove \(A_2\) is upper semicontinuous, it is sufficient to show that the graph of \(A_2\) is closed.
Let \(p_n\to p\) in the weak* topology and \(y_n\to y\) in the norm topology.  We want to show that \((x,\alpha)\in A_2(p,y) \) when \(x_n\to x\) in norm and \(\alpha_n\to \alpha\) with \((x_n,\alpha_n)\in A_2(p_n,y_n)\) for all \(n\). 
There exists \{\(f_n\)\} such that \(x_n=\int_T f_n(t)d\mu\) and \(\alpha_n = \int_T U_t(f_n(t),p_n)d\mu\) with \(f_n(t)\in B(t,p_n) \) for a.e \(t\in T\) for all \(n\). 
Clearly \(\{f_n\}\) is well-dominated. 

Let \(g_n(t) = U_t(f_n(t), p_n)\) and \(\phi_n(t) = (f_n(t), g_n(t))\). 
Then it is clear that \( \{ g_n\} \) and \( \{ \phi_n\} \) are both well-dominated. Consequently, there exists an integrable function \( \phi\) on \( T\) such that \( \phi (t) \in w\text{-Ls} \{ \phi_n (t) \} \) a.e. \( t \in T \) and \( \int_T \phi d\mu \in w\text{-Ls} \{\int_T \phi_n d\mu\}  \) by Lemma \ref{fatou}, where \( \phi(t)=(f(t), g(t)) \) for some \(f\in L^1(\mu, E)\) and \(g\in L^1(\mu)\) with \(f(t) \in X(t)\)  and \(g(t) \in \mathbb{R} \).  
Then we can extract a convergent subsequence \( \{ \phi_n \} \) (we do not relabel) such that \( \phi_n(t) \to \phi(t) \) weakly for a.e. \( t \in T \) and \(\int_T \phi_n(t)d\mu \to \int_T \phi(t)d\mu\) weakly. So we have \( f_n(t) \to f_n(t)\) weakly for a.e. \( t \in T \) , \( g_n(t) \to g(t) \) weakly for a.e. \( t \in T \) , \( \int_T f_n(t)d\mu \to \int_T f(t)d\mu\) weakly and \( \alpha_n=\int_T g_n(t)d\mu \to \int_T g(t)d\mu\).

Because \(x_n = \int_T f_n(t)d\mu\) converges to \(x\) in norm, \(\int_T f(t)d\mu =x\). By the joint continuity of \(U_t\), \(U_t(f_n(t),p_n)\to U_t(f(t),p)\)  for a.e. \( t \in T. \) Hence, we have \(g(t)= U_t(f(t),p)\) a.e. \( t \in T \) and \( \int_T U_t(f(t),p) d\mu=\int_T g(t)d\mu=\lim_n \alpha_n=\alpha  . \)

Now it remains to show \(f(t) \in B(t,p)\). %and \(\alpha = \int_T U_t(f(t),p)d\mu\).  
Because \(X(t)\) is norm compact, \(f_n(t)\) converges up to subsequence to some limit in norm, which must be equal to \(f(t)\). 
% From \(\{p_n\}\) we can also extract a subsequence, which again we do not relabel, that converges to \(p\) in the weak* topology. 
It follows that for a.e. \(t \in T\), \(p_n\cdot f_n(t) \to p\cdot f(t)\). Since \(p_n \cdot f_n(t) \leq p_n\cdot e(t) + \text{max }p_n \cdot Y(t)\), we have 
\begin{equation}
 p\cdot f(t)\leq p\cdot e(t) + \text{max }p \cdot Y(t).
\end{equation}
Therefore, \(f(t)\in B(t,p)\) for almost all \(t\in T\).  In sum, we showed that \(A_2\) is norm upper semicontinuous.

We now prove the lower semicontinuity of \(A_2\). Suppose \( (x,\alpha)\in A_2(p,y) \). In order to show \(A_2\) is lower semicontinuous, it suffices to find a sequence \( (x_n,\alpha_n) \) such that \( (x_n,\alpha_n)\in A_2(p_n,y_n) \) converging to \( (x,\alpha) \) in norm.
Since \( (x,\alpha)\in A_2(p,y) \), there exists a function \(f\) such that \( x=\int_T f(t)d\mu \) and \(\alpha = \int_T U_t(f(t),p) \). 
Notice that since for any \(p\in \Delta\), \(B(t,p)\) is a norm closed subset of \(X(t)\), it is norm compact. Clearly it is convex.  

Consider \( p_n\to p\) in the weak* topology and, \( y_n\to y \) in the norm topology. Note that \( B(t,p_n) \) is convex and norm compact. 
Thus one can choose \(f_n(t)\) from \( B(t,p_n) \) such that \(f_n(t)\) is the closest to \(f(t)\), i.e., 
\begin{equation} \label{eq:mindist}
\norm{f_n(t)-f(t)} \leq \norm{z-f(t)} \text{ for all } z \in B(t,p_n).
\end{equation} 

We will show that \(f_n\) is measurable. Note that \(B(\cdot,p)\) has a measurable graph. To see this, we adopt \cite{ky91}.
For \( p\in\Delta\), define \(\xi_p:T\times E \to [-\infty, \infty]\) by \(\xi_p(t,x)=p\cdot x-p\cdot e(t)-\text{max }p\cdot Y(t)\). By Proposition 3 in \cite{hildenbrand74} (p.60), \(\text{max }p\cdot Y(t)\) is measurable in \(t\). Then \(\xi_p\) is measurable in \(t\) and continuous in \( x \). By Proposition 3.1 in \cite{ya3}, \(\xi_p(\cdot,\cdot) \) is jointly measurable. Notice that
\begin{equation}
 G_{B(\cdot,p)}=\{(t,x)\in T\times X(t): p\cdot x \leq p\cdot e(t)+ \text{max }\; p\cdot Y(t)\}=\xi_p^{-1}([-\infty,0])\cap G_X
\end{equation}
and thus the budget correspondence \(B(\cdot,p)\) is graph measurable given \(p\).

By Castaing's Representation Theorem in \cite{ya3}, there exists \( \{h_m^n(t): m \in \mathbb{N} \} \) whose norm closure is \(B(t,p_n) \). 
Let 
\begin{equation}
 \Psi_m^n(t) = \{z\in B(t,p_n): \norm{z-f(t)}\leq \norm{h_m^n (t)-f(t)} \}
\end{equation}
and
\begin{equation}
 \Psi^n(t)\equiv \cap_{m=1}^\infty\Psi_m^n(t).
\end{equation}
From the fact that \(B(t,p)\) is norm compact and the continuity of \( \norm{\cdot} \), it follows that \(\Psi_m^n(t)\) is a non-empty measurable correspondence. Then the correspondence \(\Psi^n:T\to 2^{E}\) has a measurable graph. Since the set  \(\{h_m^n(t): m \in \mathbb{N} \}\) is dense in  \( B(t, p_n) \), only the closest point \(f_n(t)\) to \(f(t)\) belongs to \(\Psi^n(t)\). Therefore \(\Psi^n\) is a measurable function which is equal to \(f_n\) for \(\mu-\)almost all \(t \in T\).  Hence, \(f_n\) is measurable for all \(n\). It is now clear that \(f_n\in \mathcal{S}_X^1\) for all \(n\).

We will show that \(\int_T f_n (t) d\mu \to \int_T f(t) d\mu \) in norm. Let \( \varepsilon >0\). Pick \(b\in B(t,p)\cap N_{\varepsilon}(f(t)) \) where \(N_{\varepsilon}(f(t)) \) is a neighborhood of \(f(t)\) with the radius \(\varepsilon\).
Suppose \(b\notin B(t,p_n)\) for infinitely many \(n \). Then
\begin{equation}
 p_n\cdot b > p_n \cdot e(t) + \text{max } p_n\cdot Y(t).
\end{equation}

For some \(\varepsilon \in (0,1)\), we have 
\begin{equation}
 p_n \cdot \varepsilon b >  p_n \cdot e(t) + \text{max } p_n\cdot Y(t).
\end{equation}
As \( n \to \infty \), it follows 
\begin{equation}
    p \cdot \varepsilon b\geq  p \cdot e(t) + \text{max } p\cdot Y(t)
\end{equation}
which contradicts \(b\in B(t,p) \). 

Thus, there is a \( \bar n \) such that \(b \in B(t,p_n) \) for all \(n\geq\bar n \). Because of the minimizing property \eqref{eq:mindist} of \( f_n(t)\) in \( B(t, p_n ) \), we have \( \Vert f_n(t) - f(t) \Vert < \varepsilon \). So \(\lim_{n\to\infty}\int_T U_t(f_n(t),p_n)d\mu =\int_T U_t(f(t),p)d\mu\). And the Dominated Convergence Theorem \footnote{See Theorem 3 in \cite{du} p. 45.} in \cite{du} says 
\begin{equation}
\lim_{n\to\infty} \int_T \norm{f_n(t) - f(t)}d\mu = 0.
\end{equation}
Let \(x_n=\int_T f_n(t)d\mu\) and \(\alpha_n =\int_T U_t(f_n(t),p_n)d\mu\). Then \((x_n,\alpha_n)\in A_2(p_n,y_n)\) for all \(n \geq \bar n \).
Moreover, 
\begin{equation}
\norm{x_n-x} = \norm{\int_T f_n(t)d\mu - \int_T f(t)d\mu} \leq \int_T \norm{ f_n(t) - f(t) } d\mu\to 0.
\end{equation}
The last inequality comes from Theorem 4 in \cite{du} (p.46).
Hence, \(x_n\to x\) in norm and \(\alpha_n\to\alpha\).  
It follows that \(A_2\) is norm lower semicontinuous.

We will show that \(A_2\) is convex valued. 
Pick \((x,\alpha)\in A_2(p,y) \) and \( (x',\alpha')\in A_2(p,y) \). 
Then there is a function \(f:T\to E\) such that \(\int_T f(t)d\mu = x\) and \(\int_T U_t(f(t),p)d\mu = \alpha\) with \(f(t)\in B(t,p) \) a.e. \(t\) and a function \(f':T\to E \) such that \( \int_T f'(t)d\mu = x'\) and \( \int_T U_t(f'(t),p)d\mu=\alpha'\) with \(f'(t)\in B(t,p) \) a.e \(t\).   
Let \(Z =  E\times \mathbb{R}\) and we define a function \(h:T\to Z\) by \(h(t)= (f(t),U_t(f(t),p)) \) and a function \(h':T\to Z\) by \(h'(t)=(f'(t),U_t(f'(t),p)) \). 
It is clear that \(h,h'\in L^1(\mu, Z) \). 
Let \(\nu\) be a measure defined by
\begin{equation}
 \nu(S)= (\int_S h(t) d\mu, \int_S h'(t) d\mu)
\end{equation}
for \(S\in\mathcal{T}\).
Notice that \( \nu(\emptyset)=((\mathbf{0},0),(\mathbf{0},0)) \) and \( \nu(T)=((x,\alpha), (x',\alpha')) \). 
It follows from Theorem 4.1 in \cite{ks2013} (Lyapunov's convexity theorem)  
that the range of \(\nu\) is convex. Thus there exists \(S\in \mathcal{T}\) such that \( \nu(S) = \lambda \nu(T)=((\lambda x, \lambda \alpha), (\lambda x', \lambda \alpha')) \) for \(\lambda\in (0,1) \). Let \(f_\lambda = f{\chi_S}+f'{\chi_{T\backslash S}}\). Then \( \int_T f_\lambda (t) d\mu = \int_S f(t) d\mu + \int_{T\backslash S} f'(t)d\mu = \lambda x + (1-\lambda)x'\) and \( \int_S U_t(f(t),p) d\mu + \int_{T\backslash S} U_t(f'(t),p)d\mu = \lambda \alpha + (1-\lambda)\alpha'\). It is clear that \(f_\lambda (t)\in B(t,p) \). Therefore, \(A_2\) is a convex valued correspondence. 
\end{proof}

\begin{lma}
\label{gamma_eq}
   \(\Gamma\) has an equilibrium.
\end{lma}   
\begin{proof} 
As we proved in the proof of Lemma \ref{A_property}, \(K_1\), \(K_2\) and \(K_3\) are non-empty, convex and compact. 
%Note \(\int_T e(t)d\mu\in \int_T X(t)d\mu\) for all \(t\in T\) and \(\int_T U_t(e(t),p)d\mu\in [0,1]\). Thus, \(K_2\) is non-empty. By Lemma \ref{int_sy}, \(\int_T X(t)d\mu\) is norm compact and convex. It follows that \(K_2\) is compact and convex. 
Therefore, (i) of Lemma \ref{lma_debreu} is satisfied. 
Lemma \ref{A_property} shows that \(A_i\;(i=1,2,3) \) satisfies (ii) of Lemma \ref{lma_debreu}. 
It is easy to see that \(u_i\;(i=1,2,3) \) is continuous and quasi-concave on \(K_i\). Hence, (iii) of Lemma \ref{lma_debreu} holds. 
Now we can appeal to Lemma \ref{lma_debreu} to have an equilibrium \( (p^*, (x^*,\alpha^*),y^*) \) for \(\Gamma\).
\end{proof}

\section{Proof of the Auxiliary Theorem}
\label{proof_aux}

We are now ready to provide the proof of the Auxiliary Theorem. 

\vspace{5mm}
\pfa 
We will prove that for an equilibrium for \(\Gamma\), there is a competitive equilibrium for the economy. 

Suppose that \( (p^*,(x^*,\alpha^*),y^*) \) is an equilibrium for \(\Gamma\). Hence there exist \(f^*\in\mathcal{S}_X^1\) such that that \( x^*=\int_T f^*(t)d\mu \) with \(f^*(t)\in B(t,p^*)\) and \(g^*\in \mathcal{S}_Y^1 \) such that \( y^*=\int_T g^*(t)d\mu \). We will show that \((p^*,f^*,g^*) \) is a competitive equilibrium for the economy.

(i) We show that \(g^*\) is a profit maximization production plan. 

By the definition of \(u_3\), \( p^*\cdot y^*=p^*\cdot \int_T g^*(t)d\mu\geq p^*\cdot y\) for any \(y\in \int_T Y(t)d\mu\). 
Therefore, \( p^*\cdot \int_T g^*(t)d\mu = \text{max }p^{*} \cdot \int_T Y(t)d\mu\). 
By Proposition 6 in \cite{hildenbrand74} (p.63),  we have \( \text{max }p^{*} \cdot \int_T Y(t)d\mu = \int_T \text{max}\ p^*\cdot Y(t)d\mu \).  Thus \( p^*\cdot g^*(t)=\text{max }p^*\cdot Y(t)\) for almost all \(t\in T\). Note that Proposition 6 in \cite{hildenbrand74} works in our commodity space \(E\).

(ii) Let us prove \(p^*\cdot f^*(t)\leq p^*\cdot e(t) + p^*\cdot g^*(t)\) a.e. \(t\in T\). 

Note that \(f^*(t)\in B(t,p^*)=\{x\in X(t):p^*\cdot x\leq p^*\cdot e(t) + \text{max } p^*\cdot Y(t)\}\) for almost all \(t\in T\).
From \(p^*\cdot g^*(t)= \text{max }p^*\cdot Y(t)\) for a.e. \(t\in T\), we have the desired result.

(iii) We show that \(U_t(x,p^*)>U_t(f^*(t),p^*)\) implies \(p^*\cdot x> p^*\cdot e(t) + p^*\cdot g^*(t)\) for almost all \(t\in T\). 

 By way of contradiction, suppose there exists a non-empty subset \(S\in\mathcal{T}\) which is of positive measure and  let \(F\) be a correspondence from \(S\) to \(X(t)\) defined by
\(F(t)=\{x\in X(t): U_t(x,p^*)> U_t(f(t),p^*) \text{ and } p^*\cdot x \leq p^*\cdot e(t)+ p^*\cdot g^*(t)\}\) for all \(t\in S\).
Recall that \(U_t(\cdot,p^*) \) is measurable on the graph of \(X\). Recall also that \( B(\cdot,p^*) \) and \(X\) have measurable graphs. Therefore, \(F\) has a measurable graph. Moreover, since \(X\) is integrably bounded, so is \(F\). 
Hence, there is a Bochner integrable selection \(f'\) of \(F\).  
We now define \(f'' = f'{\chi_S} + f^*{\chi_{T\backslash S}}\). 
It is clear that \( \int_T U_t(f''(t),p^*) d\mu= \int_S U_t(f'(t),p^*) d\mu + \int_{T\backslash S} U_t(f^*(t),p^*) d\mu > \int_T U_t(f^*(t),p^*) d\mu =\alpha^* \) which is a contradiction.

(iv) We prove that \( (f^*,g^*) \) is a feasible allocation and a production plan. 

We know that  \( p^*\cdot f^*(t)\leq p^*\cdot e(t) + p^*\cdot g^*(t)\) a.e. \(t\in T\). 
By aggregating over \(T\), we have \( p^*\cdot (\int_T f^*(t)d\mu - \int_T e(t)d\mu -\int_T g^*(t) d\mu)\leq 0\). From the definition of the equilibrium of \(\Gamma\), it follows that for any \( p\in \Delta \),
\begin{equation}
p\cdot (\int_T f^*(t)d\mu - \int_T e(t)d\mu -\int_T g^*(t))d\mu  \leq p^*\cdot (\int_T f^*(t)d\mu - \int_T e(t)d\mu -\int_T g^*(t)d\mu)\leq 0.
\end{equation}
Therefore, \(-(\int_T f^*(t)d\mu - \int_T e(t)d\mu -\int_T g^*(t)d\mu)\in E_+\) which leads to \( \int_T f^*(t)d\mu \leq \int_T e(t)d\mu +\int_T g^*(t)d\mu \).
\qed

\section{Proof of the Main Theorem}
\label{proof_main}

We provide the proof of the Main Theorem. 
The proof follows Noguchi \cite{noguchi1997} by considering a net of truncated subeconomies, whose consumption and production sets are norm compact, which is in line with Toussaint \cite{toussaint} and Khan and Yannleis \cite{ky91}.  
From the Auxiliary Theorem, we have a net of competitive equilibria for the subeconomies. We then construct a sequence of competitive equilibria. Finally, by invoking the exact Fatou's lemma for infinite dimensional separable Banach spaces, we obtain a competitive equilibrium for the original economy.

\vspace{5mm}
\pfm 
As in Noguchi \cite{noguchi1997}, we construct the  norm compact subsets of \(X(t)\) and \(Y(t)\).    
Let \(\mathcal{F}=\{K:T\to 2^{E}| K = \text{co} (K^X\cup K^Y) \text{ where } K^X = \text{co}(\cup_{i=1}^m \varphi_i), K^Y = \text{co}(\cup_{j=1}^l \psi_j) \text{ such that }\varphi_i:T\to E \text{ and } \psi_j:T\to E \) \(\text{are } \text{measurable } \text{with } \varphi_i(t) \in X(t) \text{ and } \psi_j(t)\in Y(t) \text{ for all } t\in T; e(t), \eta(t) \in K^X(t)\text{ and } 0\in K^Y(t) \text{ for all } t\in T\}\). 

Consider \(K = \text{co}(K^X\cup K^Y)\) such that  \(K^X(t) = \text{co}(e(t)\cup \eta(t)) \text{ for all } t\in T\) and \(K^Y(t) = 0  \text{ for all }  t\in T\). Then \(K\in \mathcal{F}\) and thus \(\mathcal{F}\) is non-empty. Let \(K_1,K_2\in \mathcal{F}\). Then it is clear that \(\text{co}(K_1\cup K_2)\in \mathcal{F}\), which implies that \(\mathcal{F}\) is directed under the inclusion. 
Notice that for every \(t\in T\), \(K^X(t) = \text{co}(\cup_{i=1}^m \varphi_i(t))\) and  \(K^Y(t) = \text{co}(\cup_{j=1}^l \psi_j(t))\) are norm compact and thus \(K(t)\) is also norm compact (see Jameson \cite{jameson1974} p.208). 
Now it follows that for \(K\in \mathcal{F}\), \(K^X\) and \(K^Y\) are non-empty, convex and norm compact valued, respectively. By Theorem III. 30 in \cite{castaing_valadier},  \(K^X\) and \(K^Y\) are graph measurable. 

We define a truncated economy \(\mathcal{E}^K=[(T,\mathcal{T},\mu), (K^X(t), K^Y(t),U_t^K, e (t))_{t\in T}]\)
where \(U_t^K\) is the utility function \(U_t\) whose first domain is restricted to \(K^X(t)\). 
Since \(K^X(t)\) is convex and norm closed, by the separation theorem it is weakly closed. Thus it belongs to the Borel \(\sigma\)-
algebra generated by the weak topology of \(E\).  It is clear that \(U^K\) is measurable.

It is easy to see that \(\mathcal{E}^K\) satisfies all the assumptions of the Auxiliary Theorem. Therefore, we appeal to the Auxiliary Theorem to obtain a competitive equilibrium \((p_K,f_K, g_K) \) for \(\mathcal{E}^K\). 
Notice that \(\{(p_K,f_K, g_K):K\in \mathcal{F}\}\) is a net directed by inclusion. 
For all \(K \in \mathcal{F}\), \(K^X(t)\subset X(t)\) and, by A.1, \(X\) is integrably bounded and weakly compact valued. Thus \(\{f_K\}\) is well-dominated. We apply the same logic to \(K^Y\) and \(Y\) to see \(\{g_K\}\) is also well-dominated. %In addition, \(\{e_{k_F}\}\) is also well-dominated. 

Since \(X\) and \(Y\) are non-empty closed valued correspondences by A.1 and A.2, \((T,\mathcal{T},\mu)\) a complete probability space, \(E\) a complete separable metric space,  by Theorem III. 30 in \cite{castaing_valadier} there are two sequences of measurable functions \(\varphi_i:T\to E\) and \(\psi_i:T\to E\) such that 
\begin{equation}\label{phi_psi}
\text{cl}\{\varphi_i(t)\}_{i\in \mathbb{N}} = X(t) \text{ and } \text{cl}\{\psi_j(t)\}_{j\in \mathbb{N}} = Y(t)  \text{ for all } t\in T.
\end{equation} 
We then construct \(K_{m}^X(t)\) using \(\{\varphi_i(t)\}_{i=1}^m\) and \(K_{l}^Y(t)\) using \(\{\psi_j(t)\}_{j=1}^l\). 
Let us define \(n = \text{min }\{m,l\}\) where \(m,l\) are the numbers of  \(\varphi_i\) and of \(\psi_j\) in \(K\), respectively. 
Then consider a sequence of truncated subeconomies \(\{\mathcal{E}^n\}\) consisting of \(K_n^X(t)\) and \(K_n^Y(t)\) for all \(t\in T\). 
By the Auxiliary Theorem, we now have a sequence of competitive equilibria \( (p_n, f_n ,g_n)\) for \(\mathcal{E}^n\).

We appeal to Lemma \ref{fatou} to have \(f\in L^1(\mu,E)\) and \(g\in L^1(\mu,E)\) such that \(f(t) \in X(t)\),  \(f(t)\in w\text{-Ls}\;\{f_n(t)\}\) a.e \(t\in T\) and \( \int_T f d\mu \in w\text{-Ls}\; \{\int_T f_n d\mu\}\) as well as \(g(t)\in Y(t)\), \(g(t)\in w\text{-Ls}\;\{g_n(t)\}\) a.e. \(t\in T\)  and \(\int_T g d\mu\in w\text{-Ls}\; \{\int_T g_n d\mu\}\). Therefore, \(f\) is an allocation and \(g\) is a production plan. 
Since \(p_n\) belongs to \(\Delta\) which is weak* compact, it has a subsequence still denoted by \(p_n\) weak* converging to \( p\).

We will now show that \((p,f,g)\) is a competitive equilibrium for \(\mathcal{E}\). 

%\vspace{5mm} 
Step 1: Let us show that for \(x\in X(t)\),
\begin{equation}
	\label{inq_better}
 U_t(x,p) > U_t(f(t),p) \text{ implies } p\cdot x >  p\cdot e(t) + \text{max } p \cdot Y(t) \text{ for almost all } t \in T .
\end{equation} 
We follow Khan and Sagara \cite{ks2017} for this proof. 
By method of contradiction, suppose that there exists \(S \in  \mathcal{T}\) of positive measure with the following property: for every \(t \in S\) there exists \(\hat x \in X(t)\) such that \(U_t(\hat x,p) > U_t(f(t),p)\) and \(p \cdot \hat x \leq p \cdot e(t) + \text{max }p\cdot Y(t)\). 
Since \(p \cdot e(t)  + \text{max }p\cdot Y(t) > 0 \) by A.3 and A.8, it follows from the joint continuity of \(U_t \) that \(U_t(\varepsilon \hat x, p) > U_t(f(t),p)\) and \(p \cdot\varepsilon \hat x < p \cdot e(t) + \text{max }p\cdot Y(t)\) for some \(\varepsilon  \in  (0, 1)\). 
Thus, we can assume without loss of generality that for every \(t \in S\) 
there exists \(\hat x \in X(t)\) such that 
\(U_t(\hat x, p) > U_t(f(t), p)\) and \(p \cdot \hat x < p \cdot e(t) + \text{max }p\cdot Y(t)\). 
Let us define the correspondence \(\Lambda : S \to  2^{E}\) by 
\[
 \Lambda(t) = \{x \in X(t) \vert  U_t(x,p) > U_t(f(t),p), p \cdot x < p \cdot e(t) + \text{max }p\cdot Y(t)\}.
\]
\(\Lambda\) is an integrably bounded correspondence and \(\hat x \in \Lambda(t)\). 
We now show that \(\Lambda\) is graph measurable. 
%Since \(\Lambda(t)\) is the intersection of the correspondence defined by 
Let 
\(\Lambda_1 (t):= \{x \in X(t) \vert  U_t(x,p) > U_t(f(t),p)\}\) and \(\Lambda_2 (t) := \{x \in E \vert  p \cdot x < p \cdot e(t) + \text{max }p\cdot Y(t)\}\). Then  \(\Lambda(t)= \Lambda_1 (t)\cap \Lambda_2 (t)\). We need to prove that \(\Lambda_1\) and \(\Lambda_2\) are graph measurable. 
In the proof of Lemma \ref{A_property}, we already showed the joint measurability of the function given by \( (t, x) \to p\cdot x - p\cdot e(t) - \text{max }p\cdot Y(t) \). Therefore,  \( \Lambda_2\) is graph measurable. 

We turn to \(\Lambda_1\). 
Let \(\zeta : T\times E \to T \times E \times E\) be a mapping defined by \(\zeta (t,x) = (t,x,f(t))\) and \(\text{proj}_{T\times E\times E}\) be a projection of \((T \times E) \times (T \times E \times E)\) onto the range space \(T \times E \times E\) of \(\zeta\). 
By the projection theorem,\footnote{see Theorem III.23 in \cite{castaing_valadier}.} \(\text{proj}_{T\times E\times E}(G_\zeta)\) belongs to \(\mathcal{T}\otimes \mathcal{B}(E,w) \otimes \mathcal{B}(E,w)\).
We define a set \(H\) by 
\[
H:=\{(t,x,x')\in T \times E \times E \vert  U_t(x,p)> U_t(x',p) \} \cap ((G_X)\times E) \cap \text{proj}_{T\times E\times E}(G_\zeta).
\]
Then in view of A.5 and A.6, \(H\) belongs to \(\mathcal{T} \otimes \mathcal{B}(E,w) \otimes \mathcal{B}(E,w)\). 
Let \(\text{proj}_{T\times E}\) be the projection of \((T \times E)\times E\) onto \(T \times E\).  
Since \(G_{\Lambda_1} = \text{proj}_{T \times E} (H)\), we again appeal to the projection theorem to argue that \(G_{\Lambda_1}\) belongs to \(\mathcal{T} \otimes \mathcal{B}(E,w)\). 

Therefore, \(\Lambda\) has a measurable selection by Aumann's measurable selection theorem in \cite{aumann1969}. 
Let \(h: S \to E\) be a measurable selection from \(\Lambda\).
By Theorem III. 30 in \cite{castaing_valadier}, we can choose a sequence of measurable selections \(h_n:S\to E\) such that \(h_n(t)\in X(t)\) converges to \(h(t)\) in norm for all \(t\in S\).  
By Lemma \ref{fatou}, there exists a Bochner integrable function \(\hat h : S \to E \) such that \(\hat h(t) \in  \text{w-Ls}\{h_n(t)\} \) 
 and \(\hat h(t) \in X(t)\) a.e. \(t \in S\). 
Hence, there is a subsequence of \(\{h_n(t)\}\) in \(E\) converging weakly to \(\hat h(t)\) for a.e. \(t \in  S\).  
It is clear that \(\hat h(t) = h(t)\) a.e. \(t \in S\) and we also have \((f(t),h(t)) \in \text{w-Ls}\{(f_n(t),h_n(t))\}\) a.e. \(t \in S\).

Suppose now that the following set defined by 
\[
 \bigcup_{n\in \mathbb{N}} \{t \in  S \vert  U_t^n(h_n(t),p_n) > U_t^n(f_n(t),p_n), p_n \cdot h_n (t) < p_n \cdot e(t) + \text{max }p_n\cdot K_n^Y(t)\}
\]
is of measure zero. 
Then for each \(n\),  \(U_t^n(f_n(t),p_n)\geq U_t^n(h_n(t),p_n)\) 
or \(p_n\cdot h_n(t) \geq p_n \cdot e(t) + \text{max }p_n\cdot K_n^Y(t)\) a.e. \(t \in S\).  
Notice  for any \(y\in Y(t)\),  
there is a sequence of measurable functions \(y_n:T\to E\) such that \(y_n(t)\in K_n^Y(t)\) converges to \(y\) in norm for all \(t\in T\) by Theorem III. 30 in \cite{castaing_valadier}. 
When \(p_n\) converges to \(p\) in the weak* topology, we have \(p_n \cdot h_n(t)\to p\cdot h(t)\) and \(p_n \cdot y_n(t) \to p\cdot y\). 
Passing to the limit %along a suitable subsequence of \(\{(f_n(t), h_n(t)\}\) 
yields \(U_t(f(t),p) \geq U_t(h(t),p)\) or \(p \cdot h(t) \geq p\cdot e(t) + \text{max }p\cdot Y(t)\) a.e. \(t \in S\).\footnote{Note that for \(x\in K_n^X(t)\), \(U_t^n(x,p) = U_t(x,p)\) for any \(p\in \Delta\).}  
But this is a contradiction to the fact that \(h\) is a measurable selection from \(\Lambda\). 
Hence, there exists \(n\) such that 
\(\{t \in S \vert U_t^n(h_n(t),p_n) > U_t^n(f_n(t),p), p_n\cdot h_n(t) < p_n \cdot e(t) + \text{max }p_n\cdot K_n^Y(t)\}\) is of positive measure. 
However, this contradicts the fact that \(p_n\) and \(f_n\) are a price and an allocation of a Walrasian equilibrium for \(\mathcal{E}^n\).  We proved \eqref{inq_better}.

Indeed, we can further show that 
\begin{equation} \label{ineq:f}
 p\cdot f(t) \geq p\cdot e(t) + \text{max } p \cdot Y(t)
\end{equation}
for almost all \(t \in T\).

By A.4 (ii), if \(f(t)\) is a satiation point, \eqref{ineq:f} follows. 
If \(f(t)\) is not a satiation point, \(f(t)\) belongs to the weak closure of the upper contour set \(\{x'\in X(t): U_t(x',p)> U_t(f(t),p)\}\) for any \(p\in \Delta\). Thus, \eqref{inq_better} implies \eqref{ineq:f}.

\vspace{5mm} 
Step 2: We show that \(f\) is a feasible allocation and \( g \) is a feasible production plan. 

Since \((p_n,f_n,g_n)\) is a competitive equilibrium for \(\mathcal{E}^n\), 
it is clear that 
\( \int_T f_n(t)d\mu \leq \int_T e (t) d\mu + \int_T g_n(t) d\mu \). 
Recall that for \(\{f_n\}\) and \(\{g_n\)\} Lemma \ref{fatou} holds. Thus 
we can extract subsequences (which we do not relabel) from \(\{f_n\}\) and \(\{g_n\}\) such that \(\int_T f_n(t)d\mu \to \int_T f(t)d\mu\) weakly and \(\int_T g_n(t)d\mu \to \int_T g(t)d\mu\) weakly. 
%In addition, Lemma \ref{fatou} also works for \(\{e_n\}\), the initial endowment of \(\mathcal{E}^n\), and thus there exists \(\omega\in L^1(\mu, E)\) such that \(\omega(t) \in w\text{-Ls}\;\{e_n(t)\}\) a.e \(t\in T\) and \(\int_T \omega d\mu \in w\text{-Ls}\; \{\int_T e_n d\mu\}\). 
%And \(e_n(t) \) converges to \(e(t)\) in norm for a.e \(t\in T\). Thus \(\omega(t)=e(t)\) a.e. \(t\in T\). 
Now from \(\int_T f_n(t)d\mu \leq   \int_T e(t)d\mu + \int_T g_n(t)d\mu\) 
we obtain  
\begin{equation}\label{ineq:feas}
 \int_T f(t)d\mu \leq \int_T e(t) d\mu + \int_T g(t)d\mu. 
\end{equation}

\vspace{5mm} 
Step 3: We prove that \( p \cdot f(t) \leq p \cdot e(t) +p\cdot g(t) \) for almost all \( t \in T\). 

From \eqref{ineq:f}, we have
\begin{equation}\label{ineq:g}
p\cdot f(t)\geq p\cdot e(t) + p\cdot g(t)
\end{equation}
for almost all \(t\in T\).
By integrating \eqref{ineq:g} over \(T\), 
\begin{equation}
 \int_T [p\cdot f(t) - p\cdot e(t) - p\cdot g(t)]d\mu =  p\cdot \int_T [ f(t) - e(t) -  g(t)] d\mu \geq 0.
\end{equation}
But from \eqref{ineq:feas} it follows that 
\begin{equation}
  p\cdot \int_T [ f(t) - e(t) -  g(t)] d\mu = \int_T [p\cdot f(t) - p\cdot e(t) - p\cdot g(t)]\leq 0.
\end{equation}
Hence, we can conclude \(\int_T [p\cdot f(t) - p\cdot e(t) - p\cdot  g(t)]= 0.\)
Therefore, we have
\begin{equation}\label{eq:f}
 p\cdot f(t) = p\cdot e(t)+ p\cdot  g(t)
\end{equation}
for almost all \(t\in T\).

\vspace{5mm} 
Step 4: Let us prove \(p\cdot g(t)= \text{max } p \cdot Y(t)\) a.e. \(t\in T\). 

From \eqref{ineq:f} and \eqref{eq:f}, we have the following inequality:
\begin{equation}
\text{max } p \cdot Y(t) \leq   p\cdot f(t) - p\cdot e(t) = p\cdot  g(t)
\end{equation}
for almost all \(t\in T\). 
Obviously, we have \(\text{max } p \cdot Y(t) \geq p\cdot g(t)\). Hence, the conclusion follows. 
\qed

\section{Concluding Remarks} 
\label{conclusion}

\noindent\textbf{Remark 1 } We can appeal to 
Galerkin approximations, as suggested by Khan and Sagara \cite{ks2017}, to construct a sequence of truncated subeconomies \(\{\mathcal{E}^n\}\) with finite dimensional commodity spaces. Each \(\mathcal{E}^n\) can have a competitive equilibrium \((p_n,f_n,g_n)\) due to Greenberg et al. \cite{gsw}. We then apply the exact Fatou's lemma to the sequence of \(\{f_n,g_n\}\) to obtain \(f\) and \(g\). By the weak* compactness of \(\Delta\), we can extract a subsequence from \(\{p_n\}\) that weak* converges to \(p\in \Delta\).  Applying similar arguments as in Khan and Sagara \cite{ks2017}, we can show that \((p,f,g)\) satisfies the properties of a competitive equilibrium.\footnote{We are grateful to an anonymous referee for drawing our attention to Khan and Sagara \cite{ks2017}.} 
 %With this approach, we can obtain the existence result without the norm compact subset \(K\subset X(t)\) with \(e(t) \in K\) for all \(t\in T\).\footnote{We are grateful to an anonymous referee for drawing our attention to Khan and Sagara \cite{ks2017}.} 

\vspace{2mm}
\noindent\textbf{Remark 2 } The Auxiliary theorem can be seen as a direct proof of Greenberg et al. \cite{gsw} for infinite dimensional commodity spaces without taking the %finite dimensional 
approximation approach.  

\vspace{2mm}
\noindent\textbf{Remark 3 }
We can replace the weak compactness of production sets by the following condition: 
Let \(A_Y\) be a set defined by  
\(A_Y =\{g'\in \mathcal{S}_Y^1: \exists f'\in \mathcal{S}_X^1 \text{ s.t. } \int_T f'(t)d\mu \leq \int_T e(t)d\mu + \int_T g'(t)d\mu\}\). 
We assume \(A_Y\) is weakly compact. 

Applying our approach in the proof of the main theorem, we construct a sequence of truncated subeconomies \(\{\mathcal{E}^n\}\) and obtain a sequence of competitive equilibria \(\{p_n,f_n, g_n\}\).\footnote{Remark 3 in \cite{gsw} provided an equilibrium existence result with non-compact consumption and production sets for their economy.} 
Since \(X(t)\) is integrably bounded and weakly compact, we apply the exact Fatou's lemma to \(\{f_n\}\) to have \(f\)
and since  \(A_Y\) is weakly compact, we have \(g_n\to g\) in the weak topology. Also \(p_n \to p\in \Delta\) in the weak* topology. 
We are then able to prove that \((p,f,g)\) is a competitive equilibrium. 

For the sequence \(\{f_n\}\), we need two results: 
(by passing to a subsequence)  \(f_n\to f\) and \(f_n(t)\to f(t)\) for almost all \(t\in T\). 
These results make \(f\) a feasible allocation and \(f(t)\) a maximal element  for the agent \(t\). 
To our best knowledge, there are two ways to obtain these results: invoking the exact Fatou's Lemma or  appealing to Theorem 5.1 in Khan and Yannelis \cite{ky91}.  
Both approaches require weak compact subsets of the consumption sets. 
Therefore, even when we relax the weak compactness assumption of the consumption sets, we still need some weak compact subsets of the consumption sets which contain the set of maximal elements.


\newpage
\begin{thebibliography}{9}
\bibitem{ab}
Aliprantis, C. and K. Border (2006), 
         \emph{Infinite Dimensional Analysis}, 3rd edition, New York, Springer-Verlag.

\bibitem{auf}
Aubin, J.-P. and H. Frankowska (1990), \textit{Set-Valued Analysis}, Boston,  Birkh\"{a}user.


\bibitem{aumann1966}
Aumann, R.J. (1966), ``Existence of competitive equilibria in markets with a continuum of traders,'' \textit{Econometrica}  34, pp. 1-17. 


\bibitem{aumann1969}
Aumann, R.J. (1969), ``Measurable utility and the measurable choice theorem,'' 
         in: \textit{La D\'{e}cision}, CNRS, Aix-en-Provence, pp. 15-26.

\bibitem{balasko2003a}        
Balasko, Y. (2003), ``Economies with price-dependent preferences,'' \textit{Journal of Economic Theory} 109, pp. 333–359

\bibitem{balasko2003b} 
Balasko, Y. (2003), ``Temporary financial equilibrium,'' \textit{Economic  Theory} 21, pp. 1–18.

%\bibitem{balder2004}
%Balder, E. (2004), ``Existence of competitive equilibria in economies with a measure space of consumers and consumption externalities,'' Working Paper, University of Utrecht.

%\bibitem{balder2008}
%Balder, E. (2008), ``More on equilibria in competitive markets with externalities and a continuum of agents,'' \textit{Journal of Mathematical Economics} 44, pp. 575-602.


\bibitem{castaing_valadier}
Castaing, C. and M. Valadier (1977), \textit{Convex Analysis and Measurable Multifunctions}, Lecture Notes in Math. 580, Springer, Berlin.


%\bibitem{ct}
%Cornet, B. and M. Topuzu (2005), ``Existence of equilibria for economies with externalities and a measure space of consumers,'' \textit{Economic Theory} 26, pp. 397–421.


\bibitem{debreu}
Debreu, G. (1952), ``A social equilibrium existence theorem,''
      \textit{Proceedings of the National Academy of Sciences} 38, pp. 886-893. 


\bibitem{du}
Diestel, J. and J.J. Uhl (1977),  \textit{Vector Measures}, 
    Mathematical Surveys  No. 15. Providence, American Mathematical Society.

\bibitem{fjrdkler2002} 
Fajardo, S. and H.J. Keisler (2002), \textit{Model Theory of Stochastic Processes}, A K Peters, Ltd., Natick.

\bibitem{fremlin2012}
Fremlin, D. H. (2012), \textit{Measure Theory, Vol. 3: Measure Algebras, Part I}, second ed., Torres Fremlin, Colchester.


\bibitem{gsw}
Greenberg, J., B. Shitovitz and A. Wieczorek (1979), ``Existence of equilibria in atomless production economies with price dependent preferences,'' \textit{Journal of Mathematical Economics} 6, pp. 31-41.

\bibitem{gp2013}
Greinecker, M and K. Podczeck (2013), ``Liapounoff’s vector measure theorem in Banach spaces and applications to general equilibrium theory,'' \textit{Economic Theory Bulletin} 1, pp. 157-173.

\bibitem{gp2017}
Greinecker, M and K. Podczeck (2017), ``An Exact Fatou's Lemma for Gelfand Integrals by Means of Young Measure Theory,'' \textit{Journal of Convex Analysis} 24, pp. 621-644.

%\bibitem{hartetal}
%Hart, S., W. Hildenbrand, and E. Kohlberg (1974), ``On equilibrium allocations as distributions on the commodity space,'' \textit{Journal of Mathematical Economics} 1, pp. 159–166. 


\bibitem{hildenbrand70}
Hildenbrand, W. (1970), ``Existence of equilibria for economies with production and a measure space of consumers,'' \textit{Econometrica} 38, pp. 608-623.


\bibitem{hildenbrand74}
Hildenbrand, W. (1974), \textit{Core and Equilibria of a large economy}, 
 Princeton University Press, Princeton.
 
\bibitem{hk1984} 
Hoover, D.N. and H.J. Keisler (1984), ``Adapted probability distributions,'' 
\textit{Transactions of the American Mathematical Society} 286, pp. 159-201.

\bibitem{jameson1974}
Jameson, G. (1974), \textit{Topology and Normed Spaces},  Chapman and Hall, London.
 
\bibitem{kakutani1944} 
Kakutani, S. (1944), ``Construction of a Non-separable Extension of the Lebesgue Measure Space,'' \textit{Proceedings of the Imperial Academy} 20 (3), pp. 115-119.

\bibitem{keislersun2009}
Keisler, H.J. and Y. Sun (2009), ``Why saturated probability spaces are necessary,'' \textit{Advances in Mathematics} 221, pp. 1584–1607.


\bibitem{ks2013}
Khan, M.A. and N. Sagara (2013), ``Maharam-types and Lyapunov's theorem for vector measures on Banach spaces,''
          \textit{Illinois Journal of Mathematics} 57, pp. 145-169.


\bibitem{ks2014}
Khan, M.A. and N. Sagara (2014), ``Weak sequential convergence in \(L^1(\mu,X) \) and an exact version of Fatou's lemma,'' \textit{Journal of Mathematical Analysis and Applications} 412, pp. 554-563.

\bibitem{ks2016}
Khan, M.A. and N. Sagara (2016), ``Relaxed large economies with infinite-dimensional commodity spaces: The existence of Walrasian equilibria,'' \textit{Journal of Mathematical Economics} 67, pp. 95-107.

\bibitem{ks2017}
Khan, M.A. and N. Sagara (2017), ``Fatou's lemma, Galerkin approximations and the existence of Walrasian equilibria in infinite dimensions,'' \textit{Pure and Applied Functional Analysis} 2 (2), pp. 317-355.
 
\bibitem{kss2016}
Khan, M.A., N. Sagara and T. Suzuki (2016),  ``An exact Fatou lemma for Gelfand integrals: a characterization of the Fatou property,''\textit{Positivity} 20, pp. 343–354.
 
\bibitem{ksu2016}
Khan, M.A. and T. Suzuki (2016), ``On differentiated and indivisible commodities: An expository re-framing of Mas-Colell's 1975 model,''
          \textit{Advances in Mathematical Economics} 20, pp. 103-128.


\bibitem{ky91} 
Khan, M.A. and N.C. Yannelis (1991), 
 ``Equilibria in markets with a continuum of agents and commodities,'' in: Khan, M.A. and N.C. Yannelis, eds., \textit{Equilibrium Theory in Infinite Dimensional Spaces}, New York, Springer-Verlag, pp. 233-248.


\bibitem{lee}
Lee, S. (2013), ``Competitive equilibrium with an atomless measure space of agents and infinite dimensional commodity space without convex and complete preferences,'' \textit{Hitotsubashi Journal of Economics} 54, pp. 221-230.



\bibitem{liu}
Liu, J. (2017), ``Existence of competitive equilibrium in coalition production economies with a continuum of agents,'' \textit{International Journal of Game Theory} 46, pp. 941-955. 

\bibitem{maharam}
Maharam, D., (1942), ``On homogeneous measure algebra,'' \textit{Proceedings of the National Academy of Sciences of the USA}, 28, pp. 108-111.

\bibitem{noguchi1997}
Noguchi, M. (1997), ``Economies with a continuum of consumers, a continuum of suppliers and an infinite dimensional commodity space,'' \textit{Journal of Mathematical Economics} 27, pp. 1–21.

%\bibitem{noguchi2005}
%Noguchi, M. (2005), ``Interdependent preferences with a continuum of agents,'' \textit{Journal of Mathematical Economics} 41, pp. 665–686.

%\bibitem{nz2006}
%Noguchi, M. and W. R. Zame (2006), ``Competitive markets with externalities,'' %\textit{Theoretical Economics} 1, pp. 143-166.  

\bibitem{pod1997}
Podczeck, K. (1997), ``Markets with infinitely many commodities and a continuum of agents with non-convex preferences,'' \textit{Economic Theory} 9, pp. 385-426.


\bibitem{pod2008}
Podczeck, K. (2008), ``On the convexity and compactness of the integral of a Banach space valued correspondence,'' \textit{Journal of Mathematical Economics} 44, pp. 836-852.

%\bibitem{pod} 
%Podczeck, K. (2010), ``On existence of rich Fubini extensions,'' \textit{Economic Theory} 45, pp. 1-22.

\bibitem{shaefer_wolff}
Schaefer, H. H. and M. P. Wolff, (1999), \textit{Topological Vector Spaces}, 2nd edn., Springer, Berlin.

\bibitem{schmeidler1969}
Schmeidler, D. (1969), ``Competitive equilibria in markets with a continuum of traders and incomplete preferences,'' \textit{Econometrica} 37, pp. 578-585.

\bibitem{shafer1974}
Shafer, W. J. (1974), ``The Nontransitive Consumer,'' \textit{Econometrica} 42, pp. 913-919.


\bibitem{sy} Sun, Y. and N.C. Yannelis (2008), ``Saturation and the integration of Banach valued 
      correspondences,'' \textit{Journal of Mathematical Economics} 44, pp. 861-865.

%\bibitem{sz} Sun, Y. and Y. Zhang (2009), ``Individual risk and Lebesgue extension without aggregate uncertainty,'' \textit{Journal of Economic Theory} 144, pp. 432-443.

\bibitem{toussaint}
Toussaint, S. (1984), ``On the Existence of Equilibria in Economies with Infinitely Many Commodities and Ordered Preferences,'' \textit{Journal of Economic Theory} 33, pp. 98-115.
\bibitem{ya3} Yannelis, N.C. (1991),  ``Set-valued functions of two variables,'' in: Khan, M.A. and N.C. Yannelis, eds., \textit{Equilibrium Theory in Infinite Dimensional Spaces}, New York,  Springer-Verlag, pp. 36-72. 
\end{thebibliography}
\end{document}